\documentclass[journal]{IEEEtran}

\usepackage{amsmath}
\usepackage{amssymb}
\usepackage{graphicx} 

\usepackage{fancybox,graphicx}
\usepackage[T1]{fontenc}
\usepackage[latin1]{inputenc}
\usepackage{float}

\floatstyle{ruled}
\newfloat{alg}{tbp}{loa}
\floatname{alg}{Algorithm}

\newenvironment{algorithm}[1][\  ] %
{
\rm
\begin{tabbing} 
....\=...\=...\=...\=...\=  \+ \kill
} %
{\end{tabbing}
}

\newtheorem{definition}{Definition}[section]
\newtheorem{lemma}{Lemma}[section]
\newtheorem{corollary}{Corollary}[section]
\newtheorem{theorem}{Theorem}[section]
\newtheorem{proposition}{Proposition}[section]

\newcommand{\bx}{{\mathbf x}}

\newcommand{\by}{{\mathbf y}}
\newcommand{\bz}{{\mathbf z}}

\newcommand{\bu}{{\mathbf u}}
\newcommand{\be}{{\mathbf e}}

\newcommand{\sgn}{{\mathrm{sgn}}}

\newcommand{\supp}{{\mathrm{supp}}}
\newcommand{\Real}{{\mathbb{R}}}

\newcommand{\BlackBox}{\rule{1.5ex}{1.5ex}}  
\newenvironment{proof}{\par\noindent{\bf Proof\ }}{\hfill\BlackBox\\[2mm]}

\begin{document}

\title{Sparse Recovery with Orthogonal Matching Pursuit under RIP}
\date{}
\author{Tong~Zhang,~\IEEEmembership{Member,~IEEE,} %
\thanks{T. Zhang is with the Statistics Department,
  Rutgers University, New Jersey,  USA. E-mail: tzhang@stat.rutgers.edu. 
  The author was partially supported by the following grants: AFOSR-10097389, NSA-AMS 081024, NSF DMS-1007527, and NSF IIS-1016061.
} 
}

\maketitle

\begin{abstract}
 This paper presents a new analysis for the orthogonal matching pursuit (OMP) algorithm.
 It is shown that if the restricted isometry property (RIP) is satisfied at sparsity level $O(\bar{k})$, then OMP can stably recover a $\bar{k}$-sparse signal in 2-norm under measurement noise. 
For compressed sensing applications, this result implies that in order to uniformly
recover a $\bar{k}$-sparse signal in $\Real^d$, only $O(\bar{k} \ln d)$ random projections are needed.
This analysis improves some earlier results on OMP depending on stronger conditions that can only be satisfied with $\Omega(\bar{k}^2 \ln d)$ or $\Omega(\bar{k}^{1.6} \ln d)$ random projections.
\end{abstract}


\section{Introduction}

Consider a signal $\bar{\bx} \in \Real^d$,
and suppose that we observe its linear transformation plus measurement noise as:
\[
\by= A \bar{\bx} + \text{noise} .
\]
Here, $A$ is an $n \times d$ matrix.
If we define an objective function 
\begin{equation}
Q(\bx) = \| A \bx - \by \|_2^2 , \label{eq:quad-opt}
\end{equation}
then we may estimate the parameter $\bar{\bx}$ by minimizing $Q(\bx)$, subject to appropriate
constraints.

If $d > n$, then the solution of the unconstrained optimization problem
\begin{equation}
\min_{\bx \in \Real^d} Q(\bx) \label{eq:quad-opt-unconstr}
\end{equation}
is not unique. In order to estimate $\bar{\bx}$,
additional assumptions on $\bar{\bx}$ is necessary.
We are specifically interested in the case where $\bar{\bx}$ is sparse.
That is $\|\bar{\bx}\|_0 \ll n$, where
\[
\|x\|_0 = |\supp(x)|, \quad \supp(x)=\{j: x_j \neq 0\} .
\]
It is known that under appropriate conditions, it is possible 
to recover $\bar{\bx}$ by solving (\ref{eq:quad-opt-unconstr}) with a sparsity 
constraint as follows:
\begin{equation}
\min_{\bx \in \Real^d} Q(\bx) \qquad \text{ subject to } \|\bx\|_0 \leq k . 
\label{eq:opt}
\end{equation}
However, this optimization problem is generally NP-hard. 
Therefore one seeks computationally efficient algorithms that can approximately solve
(\ref{eq:opt}), with the goal of recovering sparse signal $\bar{\bx}$.
This paper considers the popular orthogonal matching pursuit algorithm (OMP), which has been widely 
used for this purpose (for example, see \cite{DoElTe06,Tropp04,TroppGilbert07}). We are specifically interested in two issues: 
the performance of OMP  in terms of optimizing $Q(\bx)$ and the performance of OMP in terms of
recovering the sparse signal $\bar{\bx}$.

\section{Main Result}

Our analysis considers a more general objective function $Q(\bx)$ that does not necessarily take the
quadratic form in (\ref{eq:quad-opt}). However, we assume that $Q(\bx)$ is convex.
For such a general convex objective function, we consider the fully (or totally) corrective greedy algorithm
in Figure~\ref{fig:forward-greedy}, which was analyzed in \cite{ShSrZh09}. 
This paper refines the analysis to show that the algorithm works under the restricted isometry property (RIP) of \cite{CandTao05-rip}
(the required condition will be described later in this section).
This algorithm is a direct generalization of OMP which has been traditionally considered only for 
the quadratic objective function in (\ref{eq:quad-opt}) with $F^{(0)}=\emptyset$. 
For simplicity, we assume that the number of iterations $k_0$ is chosen a priori.  
The algorithm has been known in the machine learning community
as a version of boosting \cite{WarmuthLiRa06}, and has also been proposed recently in the signal processing community \cite{BluDav08}.

\begin{figure}[ht]
\centering
 \begin{Sbox}
    \begin{minipage}{1.0\linewidth}
      \begin{algorithm}
        Input: $Q(\bx)$ defined on $\Real^d$, \+\+\+\+ \\
        initial feature set  $F^{(0)} \subset \{1,\ldots,d\}$. \-\-\-\- \\
        Output: $\bx^{(k)}$ \\
        let $\bx^{(0)}=\arg\min_{\bx \in \Real^d} Q(\bx) \quad \text{subject to } \supp(\bx) \subset F^{(0)}$ \\
        (default choice is $F^{(0)}=\emptyset$ with $\bx^{(0)}=0$) \\
        {\bf for} $k=1, 2, \ldots, k_0$ \+ \\
        let $j = \arg\max_{i} |\nabla Q(\bx^{(k-1)})_i|$ \\
        let $F^{(k)} = \{j\} \cup F^{(k-1)}$ \\
        let $\bx^{(k)}=\arg\min_{\bx \in \Real^d} Q(\bx) \quad \text{subject to } \supp(\bx) \subset F^{(k)}$ \- \\
        {\bf end}
      \end{algorithm}
    \end{minipage}
  \end{Sbox}\fbox{\TheSbox}
  \caption{Fully Corrective Greedy Boosting Algorithm (OMP)}
\label{fig:forward-greedy}
\end{figure}

For quadratic loss, the objective function $Q(\bx)$ is given by (\ref{eq:quad-opt}) and its derivative is
$\nabla Q(\bx)= 2 A^\top (A \bx - \by)$. Therefore
$j = \arg\max_{i} |\nabla Q(\bx^{(k-1)})_i|$ becomes
$j = \arg\max_{i} |{\bf a}_i^\top (A\bx-\by)|$, where ${\bf a}_i$ is the $i$-th column of matrix $A$.
This, together with $F^{(0)}=\emptyset$, leads to the standard OMP algorithm.
In order to use notation consistent with the sparse recovery literature, in the current paper,
we still refer to the more general algorithm in Figure~\ref{fig:forward-greedy} as OMP
even though it applies to objective functions other than  (\ref{eq:quad-opt}).

The general problem of optimization under sparsity constraint is NP hard. 
In order to alleviate the difficulty, we consider approximate optimization
under the restricted strong convexity assumption introduced below. 
\begin{definition}[Restricted Strong Convexity Constants]
  Given any $s \geq 0$, define restricted
  strong convexity constants $\rho_-(s)$ and $\rho_+(s)$ as follows:
  for all $\|\bx-\bx'\|_0 \leq s$, we require
  \begin{align*}
    \rho_-(s) \|\bx- \bx'\|_2^2\leq& Q(\bx') -Q(\bx) - \nabla Q(\bx)^\top (\bx'-\bx) \\
    \leq& \rho_+(s) \|\bx-\bx'\|_2^2 .
  \end{align*}
\end{definition}
If the objective function takes the quadratic form given by (\ref{eq:quad-opt}), then
the above definition is equivalent to the following sparse eigenvalue condition of $A^\top A$:
$\forall \Delta \bx \in \Real^d$ such that $\|\Delta\bx\|_0 \leq s$,
\begin{equation}
\rho_-(s) \|\Delta\bx\|_2^2 \leq \|A \Delta \bx\|_2^2
\leq \rho_+(s) \|\Delta \bx\|_2^2 . \label{eq:sparse-eig-cond}
\end{equation}
In this case, the constants $\rho_-(s)$ and $\rho_+(s)$ are closely
related to the restricted isometry constant $\delta_s$ in \cite{CandTao05-rip}, which is defined as
a constant that satisfies the condition that $\forall \Delta \bx \in \Real^d$ such that $\|\Delta\bx\|_0 \leq s$:
\[
(1-\delta_s) \|\Delta\bx\|_2^2 \leq \|A \Delta \bx\|_2^2
\leq (1+\delta_s) \|\Delta \bx\|_2^2 . 
\]
The  restricted isometry constant was
used to define the restricted isometry property (RIP) in the analysis of $L_1$ regularization method 
\cite{CandTao05-rip}.
We employ the slightly more general restricted strong convexity constants in (\ref{eq:sparse-eig-cond}) because
our analysis only requires the ratio $\rho_+(s)/\rho_-(s)$ to be bounded, and this is useful for general machine
learning problems where $\rho_+(s)$ can be larger than $2$.

In order to recover the target $\bar{\bx}$, we have to assume that 
$\bar{\bx}$ is sparse and approximately optimizes $Q(\bx)$.
If a target $\bar{\bx}$ is an exact global optimal solution, then $\nabla Q(\bar{\bx})=0$. However, this paper
deals with approximate optimal solutions, 
where $\nabla Q(\bar{\bx}) \approx 0$. In particular, we introduce the
following definition, which is convenient to apply.
\begin{definition}[Restricted Gradient Optimal Constant] \label{def:grad}
  Given $\bar{\bx} \in \Real^d$ and $s>0$, we define the restricted gradient optimal constant
  $\epsilon_s(\bar{\bx})$ as the smallest non-negative value that satisfies the following condition
  \[
  |\nabla Q(\bar{\bx})^\top \bu | \leq \epsilon_s(\bar{\bx}) \|\bu\|_2 
  \]
  for all $\bu \in \Real^d$ such that $\|\bu\|_0 \leq s$.
\end{definition}

The constant $\epsilon_s(\bar{\bx})$ measures how close is $\nabla Q(\bar{\bx})$ to zero. 
If $\nabla Q(\bar{\bx})=0$, then $\epsilon_s(\bar{\bx})=0$. 
If $\nabla Q(\bar{\bx}) \approx 0$, then $\epsilon_s(\bar{\bx})$ is small.
Moreover, similar to the definition of restricted strong convex constants, we are only interested in the value of $\nabla Q(\bar{\bx})$ in any subset of $\{1,\ldots,d\}$ with $s$ elements. The following proposition provides some estimates of
$\epsilon_s(\bar{\bx})$ using quantities that are easier to understand. 

\begin{proposition}
  We have 
  $\epsilon_s(\bar{\bx}) \leq \sqrt{s} \|\nabla Q(\bar{\bx})\|_\infty$ and
  $\epsilon_s(\bar{\bx}) \leq \|\nabla Q(\bar{\bx})\|_2$.
  Moreover, if 
  \[
Q(\bar{\bx}) \leq \inf_{\|\bx\|_0 \leq \|\bar{\bx}\|_0+s} Q(\bx) + \bar{\epsilon} ,
\]
then 
  \[
  \epsilon_s(\bar{\bx}) \leq 2 \sqrt{\rho_+(s) \bar{\epsilon} } .
  \]
  \label{prop:delta}
\end{proposition}
\begin{proof}
  The first two inequalities are straight-forward. 
  For the third inequality, we note that for $\|\bu\|_0 \leq s$:
  \begin{align*}
  &\inf_{\|\bx\|_0 \leq \|\bar{\bx}\|_0+s} Q(\bx)\\
  \leq& \inf_\eta Q(\bar{\bx} + \eta \bu) \\
  \leq& \inf_\eta [Q(\bar{\bx})+ \eta \nabla Q(\bar{\bx})^\top \bu + \rho_+(s) \eta^2 \|\bu\|_2^2 ] \\
  =& Q(\bar{\bx}) - |\nabla Q(\bar{\bx})^\top \bu|^2/ (4 \rho_+(s) \|\bu\|_2^2) .
\end{align*}
The result follows by rearranging the above inequality.
\end{proof}

The following theorem is the main result of this paper, which shows that OMP can approximately recover 
a sparse signal $\bar{\bx}$ in 2-norm if the condition (\ref{eq:rip-condition}) in the theorem
involving strong convexity constants can be satisfied. As we shall discuss later, this condition is
closely related to the RIP condition for the quadratic objective (\ref{eq:quad-opt}).
\begin{theorem}
  Consider the OMP algorithm.
  Let $\bar{\bx} \in \Real^d$ and $\bar{F}=\supp(\bar{\bx})$.
  If there exists $s$ such that 
  \begin{align}
   s \geq & |\bar{F} \cup F^{(0)}|  \nonumber \\
   & + 4 |\bar{F}\setminus F^{(0)}| 
   \frac{\rho_+(1)}{\rho_-(s)} \ln \frac{20 \rho_+(|\bar{F}\setminus F^{(0)}|)}{\rho_-(s)} , \label{eq:rip-condition}
 \end{align}
 then when $k=k_0=s-|\bar{F} \cup F^{(0)}|$, we have
  \[
  Q(\bx^{(k)}) \leq Q(\bar{\bx}) +  2.5 \epsilon_s(\bar{\bx})^2 /\rho_-(s) 
  \]
  and
  \[
  \|\bx^{(k)} - \bar{\bx}\|_2 \leq \sqrt{6} \epsilon_s(\bar{\bx}) /\rho_-(s) .
  \]
 \label{thm:omp_rip}
\end{theorem}
\begin{proof}
The detailed proof relies on a number
of technical lemmas that are left to the appendix.

The first inequality of the theorem 
is a direct consequence of Lemma~\ref{lem:omp_rip}.
The second inequality is a consequence of the first inequality and Lemma~\ref{lem:param-est}:
\begin{align*}
&\rho_-(s) \|\bx^{(k)} - \bar{\bx}\|_2^2 \\
\leq&  2    \left[Q(\bx^{(k)}) - Q(\bar{\bx}) \right] + \epsilon_s(\bar{\bx})^2 /\rho_-(s) \\
\leq& 6 \epsilon_s(\bar{\bx})^2 /\rho_-(s) .
\end{align*}
This implies the second inequality.
\end{proof}

Note that (\ref{eq:rip-condition}) can be satisfied as long as $(\rho_+(1)/\rho_-(s))\ln (\rho_+(\bar{k})/\rho_-(s))$ 
grows sub-linearly as a function of $s$. With appropriate assumptions, this allows the ratio
$\rho_+(s)/\rho_-(s)$ to be significantly larger than $1$ but bounded from above (such a condition is sometimes
referred to as sparse eigenvalue condition in the statistics literature). In this context, 
Theorem~\ref{thm:omp_rip} is useful for estimation problems encountered in machine learning,
where $\rho_+(s)/\rho_-(s)$ may be large.

In compressed sensing, one can often control the ratio of $\rho_+(s)/\rho_-(s)$ to be not much larger
than $1$ using random projection. In this context, the following result gives a simpler interpretation of the above theorem, where 
the condition (\ref{eq:rip-condition}) of the theorem is replaced by 
$\rho_+(\bar{k}) \leq 2 \rho_-(31 \bar{k})$.
\begin{corollary}
  Consider the OMP algorithm with $F^{(0)}=\emptyset$.
  Let $\bar{\bx} \in \Real^d$ and $\bar{k}=\|\bar{\bx}\|_0$.
  If the condition $\rho_+(\bar{k}) \leq 2 \rho_-(31 \bar{k})$ holds, 
  then  when $k=k_0=30 \bar{k}$, we have 
 \[
  Q(\bx^{(k)}) \leq Q(\bar{\bx}) +  2.5 \epsilon_s(\bar{\bx})^2 /\rho_-(s) 
  \]
  and
  \[
  \|\bx^{(k)} - \bar{\bx}\|_2 \leq \sqrt{6} \epsilon_s(\bar{\bx}) /\rho_-(s) ,
  \]
  where $s=31 \bar{k}$. \label{cor:omp-rip}
\end{corollary}
\begin{proof}
  If $\rho_+(\bar{k}) \leq 2 \rho_-(31 \bar{k})$ holds, 
  then we can let $s=31 \bar{k}$, which implies that 
 \[
  2 \geq \rho_+(\bar{k})/ \rho_-(s) \geq \rho_+(1)/\rho_-(s) .
  \]
  Therefore
  \begin{align*}
  s =&30 \bar{k} \geq \bar{k} + 4 \bar{k} \cdot 2 \ln (20 \cdot 2)\\
  \geq& \bar{k} + 4 \bar{k} (\rho_+(1)/\rho_-(s)) \ln (20 \rho_+(\bar{k})/\rho_-(s)) .
\end{align*}
This means that the condition (\ref{eq:rip-condition}) holds, and the corollary follows directly from
   Theorem~\ref{thm:omp_rip}.
\end{proof}

For the quadratic objective (\ref{eq:quad-opt}),
the condition $\rho_+(\bar{k}) \leq 2 \rho_-(31 \bar{k})$ is analogous to 
the RIP condition in \cite{CandTao05-rip}. 
In particular, if the matrix $A$ has the restricted isometry constant $\delta_{31\bar{k}} \leq 1/3$,
then the condition
$\rho_+(\bar{k}) \leq 4/3$ and $\rho_-(31\bar{k}) \geq 2/3$ holds, with 
$\rho_+(s)$ and $\rho_-(s)$ defined according to (\ref{eq:sparse-eig-cond}).
In this case, Corollary~\ref{cor:omp-rip} can be directly applied.

It is interesting to observe that  except for constants,
the result of this paper for OMP is as strong as those for more sophisticated greedy
algorithms such as ROMP \cite{NeeVer09-romp} or CoSaMP \cite{NeeTro08}.
For example, Corollary~\ref{cor:omp-rip} can be applied when $\delta_s \leq 1/3$ with
$s=31 \bar{k}$, while a similar result for CoSaMP in \cite{NeeTro08} applies when $\delta_s \leq 0.1$ with
$s=4 \bar{k}$. 
Nevertheless, the difference in the constants may still suggest possible advantages for more complex algorithms such as CoSaMP under suitable conditions.

For quadratic objective function, a simple instantiation of $\epsilon_s(\bar{\bx})$ using Proposition~\ref{prop:delta}
leads to the following sparse
recovery result that is relatively simple to interpret.
\begin{corollary}
  Consider the quadratic objective function $Q(\bx)= \|A \bx-\by\|_2^2$ of (\ref{eq:quad-opt}),
  and the OMP algorithm with $F^{(0)}=\emptyset$.
  Consider an arbitrary vector $\bar{\bx} \in \Real^d$ and let $\bar{k}=\|\bar{\bx}\|_0$.
  If the RIP condition $\rho_+(\bar{k}) \leq 2 \rho_-(31 \bar{k})$ holds, 
  then  when $k=k_0=30 \bar{k}$, we have 
 \[
  \|\bx^{(k)} - \bar{\bx}\|_2 \leq 2 \sqrt{6} \rho_+(s)^{1/2} \|A \bar{\bx}-\by\|_2 /\rho_-(s) ,
  \]
  where $s=31 \bar{k}$.
\label{cor:sparse-recovery-quad}
\end{corollary}

\section{Discussion}
In this paper we proved a new result for a generalization of the OMP algorithm.
It is shown that if the RIP is satisfied at sparsity level $O(\bar{k})$, then OMP can recover a $\bar{k}$-sparse signal in 2-norm. For compressed sensing applications, this result implies that
in order to uniformly recover a $\bar{k}$-sparse signal in $\Real^d$, only $n=O(\bar{k} \ln d)$ random projections are needed \cite{CandTao05-rip}. 

Our result for signal recovery is stronger than previous results for OMP that relied on different conditions. 
For example, \cite{Tropp04} considered the problem of recovering the support set of a sparse signal under
a stronger condition (also see \cite{Zhang08-forward} for recovery properties under stochastic noise).
A similar analysis was employed in \cite{TroppGilbert07},
where it was shown that for any fixed sparse signal $\bar{\bx}$ with $\bar{k}=\|\bar{\bx}\|_0$,
OMP can recover the signal with large probability
using $O(\bar{k}\ln d)$  measurements. A more refined analysis in \cite{FleRan-nips09} shows that
a lower bound of $n=2\bar{k}\ln (d-\bar{k})$ measurements is enough for recovery. 
However, the above results are not uniform with respect to all
$\bar{k}$-sparse signals $\bar{\bx}$ (that is, for any set of random projections,
there exist $\bar{k}$-sparsity signals that fail the analysis). In comparison, the RIP condition holds uniformly
by definition, and hence our result applies uniformly to all $\bar{k}$-sparse signals.
Although our result is stronger than previous results in terms of signal recovery in 2-norm, 
the result requires running the OMP algorithm for more than $\bar{k}$ iterations, and hence doesn't recover the true support set of the ideal signal. In comparison, results such as \cite{TroppGilbert07} also imply exact recovery of the correct support set (but under stronger assumptions) using only $\bar{k}$ OMP iterations.
It is also known that it is impossible to uniformly recover the support set (in $\bar{k}$ iterations) with the OMP algorithm 
with $O(\bar{k}\ln d)$  measurements \cite{Rauhut07}.
This means that it is necessary to run OMP for more than $\bar{k}$ iterations in order to achieve the best 2-norm recovery performance
with as few meausrements as possible.

It is worth mentioning that some previous results apply uniformly to all $\bar{k}$-sparse signals.
For example, results in \cite{DoElTe06} depend on the stronger mutual incoherence condition. 
Unfortunately the mutual incoherence condition can only be satisfied with $\Omega(\bar{k}^2 \ln d)$ random projections.
Therefore in recent years there have been significant interests in studying OMP under the RIP.
In addition to the current paper, a number of recent papers investigated this issue, 
reaching varying conclusions \cite{BecWoj10,DavWak10,LiuTem10-sup,Livshitz10}.
For example, the RIP-based analysis for sparse signals (but without noise) was considered in
\cite{DavWak10,LiuTem10-sup}, with the conclusion that under a sufficiently strong assumption on the RIP constant
(in fact, the resulting condition is similar to the mutual incoherence condition),
exact recovery is possible in $\bar{k}$ iterations. 
The condition required for the RIP constant was weakened in \cite{Livshitz10}, where 
the author showed that by running the OMP algorithm more than $\bar{k}$ iterations, it is possible to 
achieve exact recovery (again assuming no noise). The condition in \cite{Livshitz10}
can be satisfied with only $O(\bar{k}^{1.6} \ln d)$ measurements, which is a significant
improvement over the traditional $\Omega(\bar{k}^2 \ln d)$ measurements. 
The result obtained in the current paper is along the same line as \cite{Livshitz10},
but reduced the required number of measurements to the optimal order
of $O(\bar{k} \ln d)$.

It is also interesting to compare the new OMP result in this paper to that of Lasso, which is also known to work under the RIP. However,
a more refined comparison illustrates differences between the known theoretical results for these two methods. 
For OMP, the result in Theorem~\ref{thm:omp_rip} can be applied as long as the condition
\begin{align*}
  & s/|\bar{F} \cup F^{(0)}| \geq \\
  &\quad
   4 |\bar{F}\setminus F^{(0)}|(\rho_+(1)/\rho_-(s)) \ln (20 \rho_+(|\bar{F}\setminus F^{(0)}|)/\rho_-(s)) 
 \end{align*}
 is satisfied. With $F^{(0)}=\emptyset$, this roughly requires
$(\rho_+(1)/\rho_-(s)) \ln (\rho_+(\bar{k})/\rho_-(s))$ to grow sub-linearly as a function of $s$ in order to apply the theory.
In comparison, the known condition for Lasso (e.g., this has been made explicit in \cite{ZhangHuang06,Zhang07-l1}) requires $\rho_+(s)/\rho_-(s)$ to grow sub-linearly as a function of $s$.
To compare the two conditions, we note that the condition for OMP is weaker in terms of
of the upper convexity constant as there is no explicit dependency on $\rho_+(s)$;
however, the dependency on $\rho_-(s)$ is stronger in OMP than Lasso due to the logarithmic term.
Although it is unclear how tight these conditions are, the comparison nevertheless indicates that
even though both algorithms work under the RIP, there are still finer differences in their theoretical analysis: Lasso is slightly more favorable in terms of its dependency on the lower strong convexity constant, while OMP
is more favorable  in terms of its dependency on the upper strong convexity constant.
We further conjecture that the extra logarithmic dependency $\ln (\rho_+(\bar{k})/\rho_-(s))$ in OMP is necessary.
In practice, some times Lasso performs better while other times OMP performs better 
(for example, see experimental results in \cite{HuangZhang09:structured_sparsity}). 
Therefore some discrepancy in their theoretical analysis is expected. 
More specifically, for sparse recovery, 
one often observes that Lasso is superior when the nonzero coefficients have 
a similar magnitude (which happens to be the case that the extra $\ln (\rho_+(\bar{k})/\rho_-(s))$ factor
is required in our OMP analysis) while OMP performs better when the nonzero coefficients exhibit rapid
decay (which happens to be the case that the extra $\ln (\rho_+(\bar{k})/\rho_-(s))$ 
factor can be removed from our analysis).
The theory in this paper significantly narrows the previous theoretical gap between these two sparse recovery methods by positively answering the open question of whether OMP can recover sparse signals under the RIP. 
Therefore our result allows practitioners to apply OMP with more confidence than previously expected.

\section*{Acknowledgements}

The author would like to thank the anonymous referees for pointing out many relevant references and for
suggestions to improve the presentation.


\appendix

\section{Technical Lemmas}

We need a number of technical lemmas. 
Lemma~\ref{lem:one-step} and Lemma~\ref{lem:progress}, 
key to the proof, are based on earlier work of the author with collaborators \cite{ShSrZh09,HuangZhang09:structured_sparsity}.
The first three lemmas use the following notations.
Let $F,\bar{F}$ be two subsets of $\{1,\ldots,d\}$.
Let $\supp(\bar{\bx}) \subset \bar{F}$, and
\[
\bx = \arg\min_{\bz: \supp(\bz) \subset F} Q(\bz) .
\]

\begin{lemma} \label{lem:Q-bound}
  We have
  \[
  Q(\bx) - Q(\bar{\bx}) \leq 1.5 \rho_+(s) \| \bar{\bx}_{\bar{F}\setminus F}\|_2^2  + 
 0.5 \epsilon_s(\bar{\bx})^2/\rho_+(s) 
  \]
  for all $s \geq |\bar{F}\setminus F|$.
\end{lemma}
\begin{proof}
 Let $\bx'= \bar{\bx}_{\bar{F} \cap F}$, then by the definition of $\bx$, we know that
  $Q(\bx) \leq Q(\bx')$. Therefore
  \begin{align*}
  &Q(\bx) - Q(\bar{\bx}) \\
  \leq&
  Q(\bx') - Q(\bar{\bx}) \\
  =&
  Q(\bx') - Q(\bar{\bx}) - \nabla Q(\bar{\bx})^\top (\bx'-\bar{\bx}) 
  + \nabla Q(\bar{\bx})^\top (\bx'-\bar{\bx}) \\
  \leq & \rho_+(s) \| \bar{\bx}_{\bar{F}\setminus F}\|_2^2  + \epsilon_s(\bar{\bx}) \|\bar{\bx}_{\bar{F} \setminus F}\|_2 \\
  \leq & \rho_+(s) \| \bar{\bx}_{\bar{F}\setminus F}\|_2^2  + 0.5\epsilon_s(\bar{\bx})^2/\rho_+(s) 
  + 0.5 \rho_+(s) \|\bar{\bx}_{\bar{F} \setminus F}\|_2^2 ,
\end{align*}
which implies the lemma. The first inequality is by the definitions of $\rho_+(s)$ and $\epsilon_s(\bar{\bx})$.
The last inequality follows from the fact that 
$ab \leq 0.5 a^2 + 0.5 b^2$ with $a=\epsilon_s(\bar{\bx})/\sqrt{\rho_+(s)}$ and 
$b=\sqrt{\rho_+(s)} \|\bar{\bx}_{\bar{F} \setminus F}\|_2$.
\end{proof}

\begin{lemma} \label{lem:param-est}
We have:
\[
\rho_-(s) \|\bx - \bar{\bx}\|_2^2 \leq  2    \left[Q(\bx) - Q(\bar{\bx}) \right] + \epsilon_s(\bar{\bx})^2 /\rho_-(s) 
\]
for all $s \geq |F\cup\bar{F}|$.
\end{lemma}
\begin{proof}
From
 \begin{align*}
 &Q(\bx) - Q(\bar{\bx}) \\
  =&Q(\bx) - Q(\bar{\bx}) - \nabla Q(\bar{\bx})^\top (\bx-\bar{\bx}) + \nabla Q(\bar{\bx})^\top (\bx-\bar{\bx}) \\
  \geq& \rho_-(s)\|\bar{\bx}-\bx\|_2^2 - \epsilon_s(\bar{\bx}) \|\bar{\bx}-\bx\|_2 \\
  \geq& 0.5 \rho_-(s)\|\bar{\bx}-\bx\|_2^2 - 0.5 \epsilon_s(\bar{\bx})^2/\rho_-(s) ,
\end{align*}
we obtain the desired inequality.
The first inequality is by the definitions of $\rho_-(s)$ and $\epsilon_s(\bar{\bx})$.
The last inequality again follows from the fact that 
$ab \leq 0.5 a^2 + 0.5 b^2$ with $a=\epsilon_s(\bar{\bx})/\sqrt{\rho_+(s)}$ and 
$b=\sqrt{\rho_+(s)} \|\bar{\bx}_{\bar{F} \setminus F}\|_2$.
\end{proof}

The next lemma shows that each greedy search makes reasonable progress.
This proof is essentially identical to a similar result in
\cite{ShSrZh09} but with refined notations used in the current paper.
We thus include the proof for completeness. It allows the readers to verify more easily that the proof in \cite{ShSrZh09} remains unchanged with our new definitions.
\begin{lemma} \label{lem:one-step}
Let $\be_i \in \Real^d$ be the vector of zeros except for the $i$-th component being one. 
If $\bar{F}\setminus F \neq \emptyset$, then
for all $s \geq |F\cup\bar{F}|$:
\begin{align*}
&\min_{\alpha} Q(\bx + \alpha \be_{j}) \\
\leq& Q(\bx) -
\frac{\rho_-(s) \|\bx-\bar{\bx}\|^2}{ \rho_+(1)
\left(\sum_{i \in \bar{F}\setminus F} |\bar{\bx}_i|\right)^2} \max(0,Q(\bx)-Q(\bar{\bx})) ,
\end{align*}
where $j=  \arg\max_{i} |\nabla Q(\bx)_i|$.
\end{lemma}
\begin{proof}
For all $i \in \{1,\ldots,d\}$ and $\eta > 0 $, we define
\[
 Q_i(\eta)=Q(\bx) + \eta\,\sgn(\bar{\bx}_i)\,\nabla Q(\bx)_i + \eta^2\,\rho_+(1) .
\]
It follows from the definition of $\rho_+(1)$ that
$$
\min_{\alpha} Q(\bx + \alpha \be_{j}) ~\leq~ 
Q(\bx + \eta\,\sgn(\bar{\bx}_j)\,\be_{j}) ~\leq~ 
Q_j(\eta) .
$$
Since the choice of $j=  \arg\max_{i} |\nabla Q(\bx)_i|$ achieves the
minimum of $\min_i \min_\eta Q_i(\eta)$, the lemma is a direct consequence
of the following stronger statement:
\begin{align}
&\min_i Q_i(\eta) \label{eqn:plem:1}\\
\leq&
Q(\bx) -
\frac{\max\left(0,Q(\bx)-Q(\bar{\bx}) + \rho_-(s) \|\bx-\bar{\bx}\|^2\right)^2}{ 4 \rho_+(1)
\left(\sum_{i \in \bar{F}\setminus F} |\bar{\bx}_i|\right)^2} , \nonumber
\end{align}
with an appropriate choice of $\eta$; this is because
\begin{align*}
&\max\left(0,Q(\bx)-Q(\bar{\bx})+ \rho_-(s) \|\bx-\bar{\bx}\|^2\right)^2\\
\geq& 4 \rho_-(s)  \max(0,Q(\bx)-Q(\bar{\bx})) \|\bx-\bar{\bx}\|^2 .
\end{align*}
Therefore, we now turn to
prove that \eqref{eqn:plem:1} holds.
Denoting $u = \sum_{i \in \bar{F}\setminus F} |\bar{\bx}_i|$, we obtain that
\begin{align} \label{eqn:plem:2}
u \, \min_i Q_i(\eta)  &\leq~ 
\sum_{i\in \bar{F}\setminus F} |\bar{\bx}_i| Q_i(\eta)  \\ \nonumber
&\leq~ 
u \,Q(\bx) + \eta\,\sum_{i \in \bar{F}\setminus F} \bar{\bx}_i \, \nabla Q(\bx)_i + 
u \,\rho_+(1) \eta^2 .
\end{align}
Since we assume that $\bx$ is optimal over $F$, we get that
$\nabla Q(\bx)_i = 0$ for all $i \in F$. Additionally, $\bx_i = 0$
for $i \not\in F$ and $\bar{\bx}_i = 0$ for $i \not\in \bar{F}$. Therefore,
\begin{align*}
\sum_{i \in \bar{F}\setminus F} \bar{\bx}_i \, \nabla Q(\bx)_i &= 
\sum_{i \in \bar{F}\setminus F} (\bar{\bx}_i-\bx_i) \, \nabla Q(\bx)_i \\ 
&= 
\sum_{i \in \bar{F} \cup F} (\bar{\bx}_i-\bx_i) \, \nabla Q(\bx)_i \\
&= \nabla Q(\bx)^\top (\bar{\bx} - \bx)  ~~. 
\end{align*}
Combining the above with the definition of $\rho_-(s)$, we obtain that
\[
\sum_{i \in \bar{F}\setminus F} \bar{\bx}_i \, \nabla Q(\bx)_i 
\leq
Q(\bar{\bx})-Q(\bx) - \rho_-(s) \|\bx-\bar{\bx}\|_2^2 ~. 
\]
Combining the above with \eqref{eqn:plem:2} we get
\begin{align*}
&u\,\min_{i} Q_i(\eta) \\
\leq&
u \,Q(\bx) + \eta\,
[ Q(\bar{\bx}) - Q(\bx)- \rho_-(s) \|\bx-\bar{\bx}\|_2^2 ] 
+ u \,\rho_+(1) \eta^2 .
\end{align*}
Setting 
\[
\eta = \max[0,Q(\bx)-Q(\bar{\bx}) + \rho_-(s) \|\bx-\bar{\bx}\|_2^2 ] /(2 u \rho_+(1))
\]
and rearranging the terms, we conclude our proof of 
(\ref{eqn:plem:1}).
\end{proof}

The direct consequence of the previous lemma is the following result, which is critical
in our analysis. The idea of using a nesting approximating sequence has appeared
in \cite{HuangZhang09:structured_sparsity}, but the current version is improved. The change is necessary for the purpose of this paper. In the following $\mu$ can be chosen as any positive number if $L=1$.
\begin{lemma} \label{lem:progress}
  Consider the OMP algorithm.
  Consider a positive integer $L$ and subsets
  $\bar{F}_0 \subset \bar{F}_1 \subset \bar{F}_2 \cdots \bar{F}_L \subset \bar{F} \cup F^{(0)}$, 
  where $\bar{F}_{0}= \bar{F} \cap F^{(0)}$.
  Assume that 
  $\min_{\bx: \supp(\bx) \subset \bar{F}_j} Q(\bx) \leq Q(\bar{\bx}) + q_j$ ($j=0,\ldots,L$), $q_0 \geq q_1 \geq \cdots \geq q_{L} \geq 0$, and
  let $\mu \geq \sup_{j=1,\ldots,L-1} (q_{j-1}/q_j)$.
  If $s \geq |F^{(k)} \cup \bar{F}|$ and
  \[
  k =\sum_{j=1}^{L} \left\lceil  |\bar{F}_{j}\setminus F^{(0)}| (\rho_+(1)/\rho_-(s)) \ln (2\mu) \right\rceil ,
 \]
  then 
 \[
  Q(\bx^{(k)}) \leq Q(\bar{\bx}) + q_L + \mu^{-1} q_{L-1}.
  \]
\end{lemma}
\begin{proof} 
Note that for any $\supp(\bx) \subset F$ and $\supp(\bar{\bx}) \subset \bar{F}$, we have
when $\bar{F}\setminus F \neq \emptyset$:
\[
\frac{\rho_-(s) \|\bx-\bar{\bx}\|^2}{ \rho_+(1)
\left(\sum_{i \in \bar{F}\setminus F} |\bar{\bx}_i|\right)^2}
\geq 
\frac{\rho_-(s)}{ \rho_+(1) |\bar{F}\setminus F|} .
\]
Therefore Lemma~\ref{lem:one-step} implies that at any $k$ such that $s \geq |F^{(k)} \cup \bar{F}|$
and $\ell=0,\ldots,L$, we have either $|\bar{F}_\ell \setminus F^{(k)}|=0$ or
  \begin{align*}
 Q(\bx^{(k+1)}) \leq& Q(\bx^{(k)}) -\\
  &\frac{\rho_-(s) }{ \rho_+(1) |\bar{F}_\ell \setminus F^{(k)}|}  \max\left(0, Q(\bx^{(k)})-Q(\bar{\bx})-q_\ell\right) ,
\end{align*}
where we simply replace the target vector $\bar{\bx}$ in  Lemma~\ref{lem:one-step} by the 
  optimal solution over $\bar{F}_\ell$, and replace $\bx$ by $\bx^{(k)}$.
  The inequality, along with $Q(\bx^{(k+1)}) \leq Q(\bx^{(k)})$,
  implies that either $|\bar{F}_\ell \setminus F^{(k)}|=0$ or
  \begin{align*}
 & \max(0,Q(\bx^{(k+1)})-Q(\bar{\bx})-q_\ell) \\
~\leq&~ 
  \left[1-\frac{\rho_-(s) }{ \rho_+(1) |\bar{F}_\ell \setminus F^{(k)}|}\right]  
\max\left(0, Q(\bx^{(k)})-Q(\bar{\bx})-q_\ell\right) \\
\leq&  \exp \left[-\frac{\rho_-(s) }{ \rho_+(1) |\bar{F}_\ell \setminus F^{(k)}|}\right]  
\max\left(0, Q(\bx^{(k)})-Q(\bar{\bx})-q_\ell\right) .
\end{align*}
Therefore for any $k' \leq k$ and $\ell=1,\ldots,L$, we have
either $|\bar{F}_\ell \setminus F^{(k)}|=0$ or
\begin{align}
&  Q(\bx^{(k)})-Q(\bar{\bx})-q_\ell ~\leq~  \label{eq:proof-multi-steps} \\
&\exp \left[- \frac{\rho_-(s) (k-k')}{ \rho_+(1) |\bar{F}_\ell \setminus F^{(k')}|}\right]  
\max\left(0, Q(\bx^{(k')})-Q(\bar{\bx})-q_\ell\right) .\nonumber
\end{align}

We are now ready to prove the lemma by induction on $L$.
If $L=1$, we can set $k'=0$ in 
(\ref{eq:proof-multi-steps}) and consider any $\mu>0$.
Since
$Q(\bx^{(0)}) \leq \min_{\bx: \supp(\bx) \subset \bar{F}_0} Q(\bx) \leq Q(\bar{\bx}) + q_0$, 
we have
\[
Q(\bx^{(0)})-Q(\bar{\bx})-q_1 \leq q_0 .
\]
Therefore when 
\[
k= \left\lceil |\bar{F}_1 \setminus F^{(0)}| (\rho_+(1)/\rho_-(s)) \ln (2\mu) \right\rceil ,
\]
we have from (\ref{eq:proof-multi-steps}) that if $|\bar{F}_1\setminus \bar{F}^{(k)}| \neq 0$, then
\begin{align*}
&  Q(\bx^{(k)})-Q(\bar{\bx})-q_1 \\
~\leq&~
\exp \left[- \frac{\rho_-(s) k}{ \rho_+(1) |\bar{F}_1 \setminus F^{(0)}|}\right]  q_0\\
\leq& (2\mu)^{-1} q_0 .
\end{align*}
Note that this inequality also holds when $|\bar{F}_1\setminus \bar{F}^{(k)}|=0$, and in such case
(\ref{eq:proof-multi-steps}) does not apply. This is because in this case $Q(\bx^{(k)}) \leq \min_{\bx: \supp(\bx) \subset \bar{F}_1} Q(\bx) \leq Q(\bar{\bx})+q_1$.
Therefore the lemma always holds when $L=1$.

Now assume that the lemma holds at $L=m-1$ for some $m>1$. That is, with
\[
k' =\sum_{j=1}^{m-1}  \left\lceil |\bar{F}_{j}\setminus F^{(0)}| (\rho_+(1)/\rho_-(s)) \ln (2\mu) \right\rceil ,
\]
we have 
\[
Q(\bx^{(k')}) \leq Q(\bar{\bx}) + q_{m-1} + \mu^{-1} q_{m-2}.
\]
This implies that when $L=m$:
\[
Q(\bx^{(k')})-Q(\bar{\bx})-q_L \leq q_{L-1} + \mu^{-1} q_{L-2} - q_L
\leq 2 q_{L-1} .
\]
We thus obtain from (\ref{eq:proof-multi-steps}) that if $|\bar{F}_L\setminus \bar{F}^{(k)}| \neq 0$, then
\begin{align*}
&Q(\bx^{(k)})-Q(\bar{\bx})-q_L \\
~\leq&~ 
\exp \left[- \frac{\rho_-(s) (k-k')}{ \rho_+(1) |\bar{F}_L \setminus F^{(0)}|}\right]  (2 q_{L-1})\\
\leq& (2\mu)^{-1} (2 q_{L-1}) .
\end{align*}
Again this inequality also holds when $|\bar{F}_L\setminus \bar{F}^{(k)}|=0$, and in such case
(\ref{eq:proof-multi-steps}) does not apply. This is because in this case $Q(\bx^{(k)}) \leq \min_{\bx: \supp(\bx) \subset \bar{F}_L} Q(\bx) \leq Q(\bar{\bx})+q_L$.
This finishes the induction.
\end{proof}

The following lemma is a slightly stronger version of the theorem, which we can prove more easily by induction.
\begin{lemma}
  Consider the OMP algorithm. 
  If there exist $k$ and $s$ such that $|\bar{F}\cup F^{(k)}| \leq s$ and
\[
k = \left\lceil 4 |\bar{F}\setminus F^{(0)}| \frac{\rho_+(1)}{\rho_-(s)} \ln \frac{20 \rho_+(|\bar{F}\setminus F^{(0)}|)}{\rho_-(s)} \right\rceil ,
\]
then 
\begin{equation}
   Q(\bx^{(k)}) \leq Q(\bar{\bx}) +  2.5 \epsilon_s(\bar{\bx})^2/\rho_-(s) .
   \label{eq:proof-strong}
 \end{equation}
\label{lem:omp_rip}
\end{lemma}
\begin{proof}
We prove this result by induction on $|\bar{F}\setminus F^{(0)}|$.
If $|\bar{F}\setminus F^{(0)}|=0$, then the bound in (\ref{eq:proof-strong}) holds trivially
because $Q(\bx^{(k)}) \leq Q(\bx^{(0)}) \leq Q(\bar{\bx})$.

Assume that the claim holds with $|\bar{F}\setminus F^{(0)}| \leq m-1$ for some $m > 0$. 
Now we consider the case of $|\bar{F}\setminus F^{(0)}| = m$.
Without loss of generality, we assume for notational convenience that $\bar{F}\setminus F^{(0)}=\{1,\ldots,m\}$, and
$|\bar{\bx}_j|$ in $\bar{F}\setminus F^{(0)}$ is arranged in descending order so that
$|\bar{\bx}_1| \geq |\bar{\bx}_2| \geq \cdots \geq |\bar{\bx}_m|$.
Let $L$ be the smallest positive integer such that for all $1 \leq \ell < L$, we have
\[
\sum_{i=2^{\ell-1}}^m \bar{\bx}_i^2 < \mu
\sum_{i=2^{\ell}}^m \bar{\bx}_i^2 ,
\]
but 
\begin{equation}
\sum_{i=2^{L-1}}^m \bar{\bx}_{i}^2 \geq
\mu  \sum_{i=2^{L}}^m \bar{\bx}_i^2 , \label{eq:proof-decay}
\end{equation}
where $\mu=10\rho_+(m)/\rho_-(s)$.
We have $L \leq \lfloor \log_2 m \rfloor+1$ because the second inequality
is automatically satisfied when $L = \lfloor \log_2 m \rfloor+1$
(the right hand side is zero in this case).
Moreover, if the second inequality is always satisfied for all $L \geq 1$, then we can simply take $L=1$ (and ignore the
first inequality).

We can now define
\[
\bar{F}_\ell = (\bar{F} \cap F^{(0)}) \cup \{i: 1 \leq i \leq \min(m,2^{\ell}-1)\} 
\]
for $\ell=0,1,2 ,\ldots, L$.

Lemma~\ref{lem:Q-bound} implies that for $\ell=0,1,\ldots,L$:
\begin{align*}
&\min_{\bx \subset \bar{F}_\ell} Q(\bx) \leq Q(\bar{\bx}) + q_\ell ,\\
&q_\ell =  1.5 \rho_+(m) \sum_{i=2^{\ell}}^m \bar{\bx}_i^2  +  0.5 \epsilon_s(\bar{\bx})^2 /\rho_+(m) . 
\end{align*}
Moreover $q_{\ell-1} \leq \mu q_\ell$ when $\ell=1,\ldots,L-1$.
We can thus apply Lemma~\ref{lem:progress} to conclude that when 
\begin{align}
k =&\sum_{j=1}^{L} \left\lceil  (2^j-1) (\rho_+(1)/\rho_-(s)) \ln (2\mu) \right\rceil \nonumber\\
\leq&   2^{L+1} (\rho_+(1)/\rho_-(s)) \ln (2\mu) -1 ,
\label{eq:proof-ind-k}
\end{align}
we have
\begin{align}
&Q(\bx^{(k)}) - Q(\bar{\bx}) \nonumber\\
\leq &
1.5 \rho_+(m) \sum_{i=2^{L}}^m \bar{\bx}_i^2  
+ 1.5 \mu^{-1} \rho_+(m) \sum_{i=2^{L-1}}^m \bar{\bx}_i^2  \nonumber\\
&\quad
+ 0.5 (1+\mu^{-1}) \epsilon_s(\bar{\bx})^2 /\rho_+(m)  \nonumber\\
\leq& 3 \mu^{-1} \rho_+(m) \sum_{i=2^{L-1}}^m \bar{\bx}_i^2  
+  \frac{0.5}{\rho_+(m)} (1+\mu^{-1}) \epsilon_s(\bar{\bx})^2   , \label{eq:proof-barx-progress}
\end{align}
where (\ref{eq:proof-decay}) is used to derive the second
inequality.

Now, if 
\begin{equation}
 2\mu^{-1} \rho_+(m)\sum_{i=2^{L-1}}^m \bar{\bx}_i^2   \leq (1+\mu^{-1}) \epsilon_s(\bar{\bx})^2 /\rho_-(s) ,
\label{eq:proof-barx-small}
\end{equation}
then (\ref{eq:proof-barx-progress}) implies that (\ref{eq:proof-strong}) holds automatically
(since $\mu \geq 10$), which finishes the induction.
Therefore in the following, we only consider the case (\ref{eq:proof-barx-small}) does not hold, which implies that
\[
 2\mu^{-1} \rho_+(m)\sum_{i=2^{L-1}}^m \bar{\bx}_i^2  
 > (1+\mu^{-1}) \epsilon_s(\bar{\bx})^2 /\rho_-(s) .
\]
Now Lemma~\ref{lem:param-est} implies that 
\begin{align*}
&\rho_-(s) \|\bx^{(k)} - \bar{\bx}\|_2^2\\
\leq& 2 (Q(\bx^{(k)}) - Q(\bar{\bx})) + \epsilon_s(\bar{\bx})^2/\rho_-(s)\\
\leq& 6 \mu^{-1} \rho_+(m) \sum_{i=2^{L-1}}^m \bar{\bx}_i^2  
+ (2+\mu^{-1}) \epsilon_s(\bar{\bx})^2/\rho_-(s)\\
<&
10 \mu^{-1} \rho_+(m) \sum_{i=2^{L-1}}^m \bar{\bx}_i^2  
=\rho_-(s) \sum_{i=2^{L-1}}^m \bar{\bx}_i^2  .
\end{align*}
This implies that
\[
\sum_{i=m-|\bar{F}\setminus F^{(k)}|+1}^m \bar{\bx}_i^2
\leq \sum_{ i \in \bar{F}\setminus F^{(k)}} \bar{\bx}_i^2
\leq \|\bx^{(k)} - \bar{\bx}\|_2^2
< \sum_{i=2^{L-1}}^m \bar{\bx}_i^2 .
\]
Therefore $m-|\bar{F}\setminus F^{(k)}|+1> 2^{L-1}$. That is,
$|\bar{F}\setminus F^{(k)}| \leq m - 2^{L-1}$.
It follows from the induction hypothesis that after another 
\[
\lceil 4 (m-2^{L-1}) (\rho_+(1)/\rho_-(s)) \ln (2\mu) \rceil
\]
OMP iterations, (\ref{eq:proof-strong}) holds.
Therefore by combining this estimate
with (\ref{eq:proof-ind-k}), we know that 
the total number of OMP iterations for (\ref{eq:proof-strong}) to hold (starting with $F^{(0)}$) is no more than
\begin{align*}
&\lceil 4 (m-2^{L-1})  (\rho_+(1)/\rho_-(s)) \ln (2\mu) \rceil \\
& \quad +  2^{L+1} (\rho_+(1)/\rho_-(s)) \ln (2\mu)  -1 \\
\leq & \lceil 4 m  (\rho_+(1)/\rho_-(s)) \ln (2\mu) \rceil .
\end{align*}
This finishes the induction step for the case $|\bar{F}\setminus F^{(0)}|=m$.
\end{proof}

\begin{IEEEbiographynophoto}{Tong Zhang}
Tong Zhang received a B.A. in mathematics and computer science from Cornell University in 1994 and a Ph.D. in Computer Science from Stanford University in 1998.
 After graduation, he worked at IBM T.J. Watson Research Center in Yorktown Heights, New York, and Yahoo Research in New York city. He is currently a professor of statistics at Rutgers University.
His research interests include machine learning, algorithms for statistical computation, their mathematical analysis and applications.
\end{IEEEbiographynophoto}

\end{document}